\documentclass[a4paper,USenglish,numberwithinsect]{lipics-v2016}
\usepackage{microtype}%if unwanted, comment out or use option "draft"

\usepackage{polski} % Displaying Polish names

\usepackage{thmtools, thm-restate} % Restating a theorem with the same number

\usepackage{complexity}
\usepackage{algorithm} % Boxes/formatting around algorithms
\usepackage{algorithmicx}
\usepackage[noend]{algpseudocode}

\theoremstyle{plain}
\newtheorem{proposition}[theorem]{Proposition}
\newtheorem{claim}[theorem]{Claim}
\newtheorem{openquestion}[theorem]{Open Question}

\newlang{\SteinerTree}{Steiner \ Tree}
\newlang{\SteinerForest}{Steiner \ Forest}
\newlang{\DirectedSteinerTree}{Directed \ Steiner \ Tree}
\newlang{\DST}{DST}
\newlang{\DirectedSteinerNetwork}{Directed \ Steiner \ Network}
\newlang{\DSN}{DSN}
\newlang{\PrioritySteinerTree}{Priority \ Steiner \ Tree}

\newlang{\TemporalSteinerNetwork}{Temporal \ Steiner \ Network}
\newlang{\kTemporalSteinerNetwork}{\mathit{k}\textsf{-}Temporal \ Steiner \ Network}
\newlang{\TSN}{TSN}
\newlang{\kTSN}{\mathit{k}\textsf{-}TSN}
\newlang{\DirectedTemporalSteinerNetwork}{Directed \ Temporal \ Steiner \ Network}
\newlang{\kDirectedTemporalSteinerNetwork}{\mathit{k}\textsf{-}Directed \ Temporal \ Steiner \ Network}
\newlang{\DTSN}{DTSN}
\newlang{\kDTSN}{\mathit{k}\textsf{-}DTSN}

\newlang{\Simple}{Simple}
\newlang{\Monotonic}{Monotonic}
\newlang{\SingleSource}{Single\textsf{-}Source}

\newlang{\LabelCover}{Label \ Cover}
\newlang{\LC}{LC}
\newlang{\kPartiteHypergraphLabelCover}{\mathit{k}\textsf{-}Partite \ Hypergraph \ Label \ Cover}
\newlang{\kPHLC}{\mathit{k}\textsf{-}PHLC}
\title{Steiner Network Problems on Temporal Graphs\footnote{This work was partially supported by the National Science Foundation Graduate Research Fellowship Program award DGE 1106400, NIH grants U01HG007910 and U01MH105979, and the U.S.-Israel Binational Science Foundation.}}
\author[1]{Alex Khodaverdian}
\author[2]{Benjamin Weitz}
\author[3]{Jimmy Wu}
\author[4]{Nir Yosef}
\affil[1]{UC Berkeley, Berkeley, USA \\
  \texttt{alexkhodaverdian@berkeley.edu}}
\affil[2]{UC Berkeley, Berkeley, USA \\
  \texttt{bsweitz@cs.berkeley.edu}}
\affil[3]{Stanford University, Stanford, CA, USA \\
  \texttt{jimmyjwu@stanford.edu}}
\affil[4]{UC Berkeley, Berkeley, USA \\
  \texttt{niryosef@eecs.berkeley.edu}}
\authorrunning{A. Khodaverdian, B. Weitz, J. Wu, and N. Yosef} %mandatory. First: Use abbreviated first/middle names. Second (only in severe cases): Use first author plus 'et. al.'

\Copyright{Alex Khodaverdian, Benjamin Weitz, Jimmy Wu, and Nir Yosef}%mandatory, please use full first names. LIPIcs license is "CC-BY";  http://creativecommons.org/licenses/by/3.0/

% mandatory: Please choose ACM 1998 classifications from http://www.acm.org/about/class/ccs98-html . E.g., cite as "F.1.1 Models of Computation". 
% mandatory: Please provide 1-5 keywords
\EventEditors{John Q. Open and Joan R. Acces}
\EventNoEds{2}
\EventLongTitle{42nd Conference on Very Important Topics (CVIT 2016)}
\EventShortTitle{CVIT 2016}
\EventAcronym{CVIT}
\EventYear{2016}
\EventDate{December 24--27, 2016}
\EventLocation{Little Whinging, United Kingdom}
\EventLogo{}
\SeriesVolume{42}
\ArticleNo{23}
\begin{document}

\maketitle

\begin{abstract}
We introduce a \emph{temporal} Steiner network problem in which a graph, as well as changes to its edges and/or vertices over a set of discrete times, are given as input; the goal is to find a minimal subgraph satisfying a set of $k$ time-sensitive connectivity demands. We show that this problem, $\kTemporalSteinerNetwork$ ($\kTSN$), is $\NP$-hard to approximate to a factor of $k - \epsilon$, for every fixed $k \geq 2$ and $\epsilon > 0$. This bound is tight, as certified by a trivial approximation algorithm. Conceptually this demonstrates, in contrast to known results for traditional Steiner problems, that a time dimension adds considerable complexity even when the problem is offline.

We also discuss special cases of $\kTSN$ in which the graph changes satisfy a \emph{monotonicity} property. We show approximation-preserving reductions from monotonic $\kTSN$ to well-studied problems such as $\PrioritySteinerTree$ and $\DirectedSteinerTree$, implying improved approximation algorithms.

Lastly, $\kTSN$ and its variants arise naturally in computational biology; to facilitate such applications, we devise an integer linear program for $\kTSN$ based on network flows.
\end{abstract}

%----------------------------------------------------------------------------------------
%    INTRODUCTION
%----------------------------------------------------------------------------------------
\section{Introduction}
The $\SteinerTree$ problem, along with its many variants and generalizations, is a core family of $\NP$-hard combinatorial optimization problems. Like many such problems, they have been intensely studied in both the classic ``static'' setting in which a single graph is given as input up-front, as well as the online and dynamic settings, in which an algorithm is required to produce outputs or decisions as parts of the input arrive. In this paper, we offer a perspective that sits in between the static and the online cases; in our \emph{temporal} setting, a graph, as well as changes to it over a set of discrete times, are all given as immediate input. Our study of this problem draws motivation and techniques from several lines of research, which we briefly summarize.

%----------------------------------------------------------------------------------------
\subparagraph*{Classic Steiner problems}
The most basic Steiner problem is $\SteinerTree$: given a weighted undirected graph $G=(V,E)$ and a set of terminals $T \subseteq V$, find a minimum-weight subtree that spans $T$. The $\SteinerForest$ problem generalizes this: given $G=(V,E)$ and a set of demand pairs $D \subseteq V \times V$, find a subgraph that connects each pair in $D$. Currently the best known approximation algorithms give a ratio of 1.39 for $\SteinerTree$ \cite{byrka2010improved} and 2 for $\SteinerForest$ \cite{agrawal1995trees}. These problems are known to be $\NP$-hard to approximate to within some small constant \cite{chlebik2008steiner}.

For directed graphs, we have the $\DirectedSteinerNetwork$ problem ($\DSN$), in which we are given a weighted directed graph $G=(V,E)$ and $k$ demands $(a_1,b_1), \ldots, (a_k,b_k) \in V \times V$, and must find a minimum-weight subgraph in which each $a_i$ has a path to $b_i$. When $k$ is fixed, $\DSN$ admits a polynomial-time exact algorithm \cite{feldman1999directed}. For general $k$, the best known approximation algorithms have ratio $O(k^{1/2 + \epsilon})$ for any fixed $\epsilon > 0$ \cite{feldman2009improved,chekuri2011set}. On the complexity side, Dodis and Khanna \cite{dodis1999design} ruled out a polynomial-time $O(2^{\log^{1-\epsilon} n})$-approximation unless $\NP$ has quasipolynomial-time algorithms. An important special case of $\DSN$ is $\DirectedSteinerTree$, in which all demands have the form $(r,b_i)$ for some root node $r$. This problem has a $O(k^\epsilon)$-approximation scheme \cite{charikar1999approximation} and a lower bound of $\Omega(\log^{2-\epsilon} n)$ \cite{halperin2003polylogarithmic}.

Another related problem which will be useful to us is $\PrioritySteinerTree$, in which each edge has an associated \emph{priority} value, and each demand must be routed using edges of at least a certain priority. Charikar, Naor, and Schieber introduced this problem and gave a $O(\log k)$-approximation \cite{charikar2004resource}.

\subparagraph*{Network Biology}
A central object of study in molecular biology is a \emph{protein-protein interaction} (PPI) network, in which nodes represent proteins and each edge $(u,v)$ represents a physical interaction between proteins $u$ and $v$, occurring with some probability $p(u,v)$. Thus by a chain of such interactions, a signal can be propagated through the network from one protein $a$ to another protein $b$ with which it has no direct contact.

A common setting is that from laboratory experiments, it is known that when a particular biological process occurs, for a certain collection of protein pairs $\{(a_1,b_1),\ldots,(a_k,b_k)\}$, there must be a chain of interactions between each $a_i$ and $b_i$, though it is not known which intermediate proteins participated in these interaction chains. To infer these intermediate proteins computationally, we can try to find a subgraph, of maximum joint probability, that simultaneously enables signals between all the protein pairs, thereby explaining the overall biological activity. Setting the edges to have weights $w(e) = -\log p(e)$, the task becomes minimizing the total edge weight---precisely the $\SteinerForest$ problem. Indeed, variations on this idea have been used effectively by \cite{scott2005identifying}, \cite{huang2009integration}, \cite{yosef2009toward}, \cite{BenSh2012yeast}, \cite{Wu2013tcell}, and others to understand signal transduction pathways in living cells.

Despite these successes, most of the existing literature ignores a critical dimension of the problem: in reality, proteins can become activated or inactivated over time, forming a dynamic PPI network \cite{PrzytyckaSS10}. An important research direction is therefore to study Steiner problems that incorporate the times at which each demand should be satisfied; the work presented here is motivated by this challenge. See \autoref{section:computational_biology} for more details. \\

In this paper, we introduce temporal generalizations of $\DirectedSteinerNetwork$ and $\SteinerForest$. The main variant we consider is the following:

\begin{quote}
$\kTemporalSteinerNetwork$ ($\kTSN$) \\
\textbf{Input:}
\begin{enumerate}
    \item A sequence of undirected graphs or \emph{frames} $G_1 = (V,E_1), G_2 = (V,E_2), \ldots, G_T = (V,E_T)$ on the same vertex set. Each edge $e$ in the underlying edge set $E := \bigcup_t E_t$ has weight $w(e) \geq 0$.
    \item A set of $k$ \emph{connectivity demands} $\mathcal{D} \subseteq V \times V \times [T]$.
\end{enumerate}

We call $\mathcal{G} = (V,E)$ the \emph{underlying graph}. We say a subgraph $\mathcal{H} \subseteq \mathcal{G}$ \emph{satisfies} demand $(a,b,t) \in \mathcal{D}$ iff $\mathcal{H}$ contains an $a$-$b$ path $P$ along which all edges exist in $G_t$.

\textbf{Output:} A minimum-weight subgraph $\mathcal{H} \subseteq \mathcal{G}$ that satisfies every demand in $\mathcal{D}$.
\end{quote}

Similarly $\kDirectedTemporalSteinerNetwork$ ($\kDTSN$) is the same problem except that the edges are directed, and a demand $(a,b,t)$ must be satisfied by a directed path from $a$ to $b$ in $G_t$.

% In addition to the molecular biology application described before, $\kTSN$ and $\kDTSN$ are also natural models for scenarios in which a network designer wishes to purchase links, routers, or other hardware or services, so as to satisfy a set of time-sensitive communication requirements.

%----------------------------------------------------------------------------------------
\subsection{Our Results}
In \autoref{section:hardness_of_temporal_steiner_problems}, we show a strong inapproximability result for $\kTemporalSteinerNetwork$ ($\kTSN$) and its directed version $\kDTSN$:

\begin{restatable}[Main Theorem]{theorem}{inapproximabilityforgeneralk} \label{theorem:inapproximability_for_general_k}
$\kTSN$ and $\kDTSN$ are $\NP$-hard to approximate to a factor of $k - \epsilon$ for every $k \geq 2$ and every constant $\epsilon > 0$. For $\kDTSN$, this holds even when the underlying graph is acyclic.
\end{restatable}

Thus the best approximation ratio one can hope for is $k$, which is easily achieved by taking the union of shortest paths for each demand. This contrasts with the static Steiner network problems, which have nontrivial approximation algorithms and efficient fixed-parameter algorithms. Our proof is via a reduction from Feige's $k$-prover system \cite{feige1998threshold}, which can be viewed as a $\LabelCover$ problem on partite hypergraphs.

In \autoref{section:monotonic_special_cases}, we discuss a broad class of special cases in which the edges are \emph{monotonic}: once an edge exists, it exists for all future times. We observe that monotonic $\kTSN$ is essentially equivalent to the well-studied $\PrioritySteinerTree$ problem, and inherits the approximability bounds for that problem.

\begin{theorem} \label{theorem:monotonic_kTSN_approximability}
$\Monotonic$ $\kTSN$ has a polynomial-time $O(\log k)$-approximation algorithm. It has no $\Omega(\log \log n)$-approximation algorithm unless $\NP \in \DTIME(n^{\log \log \log n})$.
\end{theorem}

For monotonic $\kDTSN$ with a single source (that is, every demand is of the form $(r,b,t)$ for a common root node $r$), we show the following:

\begin{restatable}{theorem}{monotonicsinglesourcekDTSNapproximability} \label{theorem:monotonic_singlesource_kDTSN_approximability}
$\Monotonic$ $\SingleSource$ $\kDTSN$ has a polynomial-time $O(k^\epsilon)$-approximation algorithm for every $\epsilon>0$. It has no $\Omega(\log^{2-\epsilon} n)$-approximation algorithm unless $\NP \in \ZPTIME(n^{\polylog(n)})$.
\end{restatable}

This is achieved via approximation-preserving reductions to and from $\DirectedSteinerTree$, which is known to have an $O(k^\epsilon)$-approximation scheme~\cite{charikar1999approximation} and a $\Omega(\log^{2-\epsilon} n)$ lower bound~\cite{halperin2003polylogarithmic}. We also devise an explicit approximation algorithm achieving the same upper bound.

Finally for the general $\kDTSN$ problem, we present an integer linear program in \autoref{section:exact_algorithm}. We make an implementation of this publicly available at \url{https://github.com/YosefLab/dynamic_connectivity}.
%----------------------------------------------------------------------------------------
%----------------------------------------------------------------------------------------

%----------------------------------------------------------------------------------------
%    PRELIMINARIES
%----------------------------------------------------------------------------------------
\section{Preliminaries}

Note that the formulations of $\kTSN$ and $\kDTSN$ in the Introduction involved a fixed vertex set; only the edges change over the times or frames. It is also natural to formulate the temporal Steiner network problem with nodes changing over time, or both nodes and edges. However by the following proposition, it is no loss of generality to discuss only the edge-temporal variant.

\begin{restatable}{proposition}{variantsmutuallyreducible} \label{proposition:variants_mutually_reducible}
The edge, node, and node-and-edge variants of $\kTSN$ are mutually polynomial-time reducible via strict reductions (i.e. preserving the approximation ratio exactly). Similarly all three variants of $\kDTSN$ are mutually strictly reducible.
\end{restatable}

We defer the precise definitions of the other two variants, as well as the proof of this proposition, to \autoref{section:problem_variants}.

Next we state the $\LabelCover$ problem, which is the starting point of one of our reductions to $\kTSN$.

\begin{definition}[$\LabelCover$ ($\LC$)]
An instance of this problem consists of a bipartite graph $G = (U,V,E)$ and a set of possible labels $\Sigma$. The input also includes, for each edge $(u,v) \in E$, \emph{projection functions} $\pi_u^{(u,v)} : \Sigma \to C$ and $\pi_v^{(u,v)} : \Sigma \to C$, where $C$ is a common set of colors; $\Pi = \{\pi_v^e : e \in E, v \in e \}$ is the set of all such functions. A labeling of $G$ is a function $\phi : U \cup V \to \Sigma$ assigning each node a label. We say a labeling $\phi$ \emph{satisfies} an edge $(u,v) \in E$, or $(u,v)$ is \emph{consistent} under $\phi$, iff $\pi_u^{(u,v)}\left(\phi(u)\right) = \pi_v^{(u,v)}\left(\phi(v)\right)$. The task is to find a labeling that satisfies as many edges as possible.
\end{definition}

This slightly generalizes the original definition in \cite{arora1993hardness}. It has the following gap hardness, which follows by combining the PCP theorem \cite{arora1998proof} with Raz's parallel repetition theorem \cite{raz1998parallel}.

\begin{theorem}\label{theorem:label_cover_hardness}
For every $\epsilon > 0$, there is a constant $|\Sigma|$ such that the following promise problem is $\NP$-hard: Given a $\LabelCover$ instance $(G, \Sigma, \Pi)$, distinguish between the following cases:
\begin{itemize}
	\item (YES instance) There exists a \emph{total labeling} of $G$; i.e. a labeling that satisfies every edge.
    \item (NO instance) There does not exist a labeling of $G$ that satisfies more than $\epsilon |E|$ edges.
\end{itemize}
\end{theorem}

In \autoref{section:hardness_of_temporal_steiner_problems}, we use $\LabelCover$ to show $(2 - \epsilon)$-hardness for 2-$\TSN$ and 2-$\DTSN$; that is, when there are only two demands. To prove our main result however, we will actually need a generalization of $\LabelCover$ to partite hypergraphs, called $\kPartiteHypergraphLabelCover$. Out of space considerations we defer the statement of this problem and its gap hardness to \autoref{section:inapproximability_for_general_k}, where the $(2 - \epsilon)$-hardness result is generalized to show $(k - \epsilon)$-hardness for general number of demands $k$.
%----------------------------------------------------------------------------------------
%----------------------------------------------------------------------------------------

%----------------------------------------------------------------------------------------
%    HARDNESS OF TEMPORAL STEINER PROBLEMS
%----------------------------------------------------------------------------------------
\section{Hardness of Temporal Steiner Problems} \label{section:hardness_of_temporal_steiner_problems}

%----------------------------------------------------------------------------------------
\subsection{Overview of the Reduction}
We first outline our strategy for reducing $\LabelCover$ to the temporal Steiner problems; specifically, we reduce to 2-$\TSN$. A similar hardness for $\kTSN$ is obtained by using the same ideas, but reducing from $\kPartiteHypergraphLabelCover$.

Consider the nodes $u_1, \ldots, u_{|U|}$ on the ``left'' side of the $\LC$ instance. We build, for each $u_i$, a gadget (which is a small subgraph in the Steiner instance) consisting of multiple parallel directed paths from a source to a sink---one path for each possible label for $u_i$. We then chain together these gadgets, so that the sink of $u_1$'s gadget is the source of $u_2$'s gadget, and so forth. Finally we create a connectivity demand from the source of $u_1$'s gadget to the sink of $u_{|U|}$'s gadget, so that a solution to the Steiner instance must have a path from $u_1$'s gadget, through all the other gadgets, and finally ending at $u_{|U|}$'s gadget. This path, depending on which of the parallel paths it takes through each gadget, induces a labeling of the left side of the $\LC$ instance. We build an analogous chain of gadgets for the nodes on the right side of the $\LC$ instance.

The last piece of the construction is to ensure that the Steiner instance has a low-cost solution if and only if the $\LC$ instance has a consistent labeling. This is accomplished by setting all the $u_i$ gadgets to exist only at time 1 (i.e. in frame $G_1$), setting the $v_j$ gadgets to exist only in $G_2$, and then merging certain edges from the $u_i$-gadgets with edges from the $v_j$-gadgets, replacing them with a single, shared edge that exists in both frames. Intuitively, the edges we merge are from paths that correspond to labels that satisfy the $\LabelCover$ edge constraints. The result is that a YES instance of $\LC$ (i.e. one with a total labeling) will enable a high degree of overlap between paths in the Steiner instance, so that there is a very low-cost solution. On the other hand, a NO instance of $\LC$ will not result in much overlap between the Steiner gadgets, so every solution will be costly.

Let us define some of the building blocks of the reduction we just sketched:
\begin{itemize}
    \item A \emph{bundle} is a graph gadget consisting of a source node $b_1$, sink node $b_2$, and parallel, disjoint \emph{strands} (defined shortly) from $b_1$ to $b_2$.
    \item A \emph{chain} of bundles is a sequence of bundles, with the sink of one bundle serving as the source of another.
    \item A \emph{simple strand} is a directed path of the form $b_1 \to c_1 \to c_2 \to b_2$.
    \item In a simple strand, we say that $(c_1, c_2)$ is the \emph{contact edge}. Contact edges have weight 1; all other edges in our construction have zero weight.
    \item More generally, a \emph{strand} can be made more complicated, by replacing a contact edge with another bundle (or even a chain of them). In this way, bundles can be nested, as shown in \autoref{fig:bundle_example}.
    \item We can \emph{merge} two or more simple strands from different bundles by setting their contact edges to be the same edge, and making that edge existent at the union of all times when the original edges existed (\autoref{fig:merge_example}).
\end{itemize}

\begin{figure}[H]
\centering
\begin{minipage}{.4\textwidth}
	\centering
    \includegraphics[width=\linewidth]{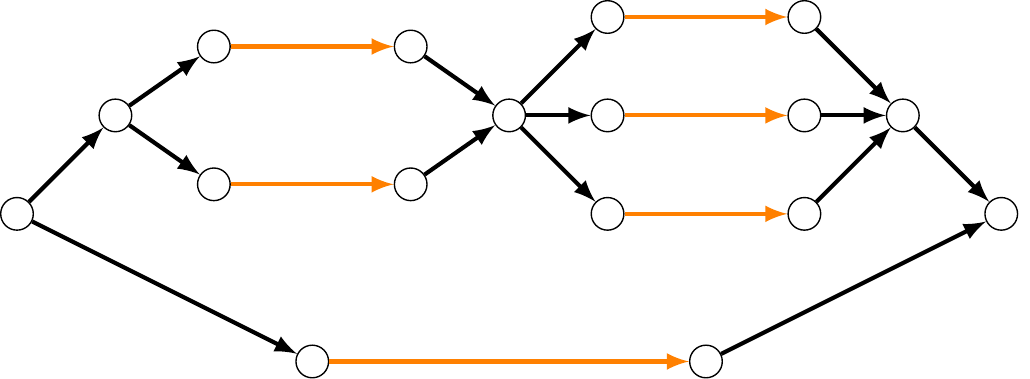}
	\captionof{figure}{A bundle whose upper strand is a chain of two bundles; the lower strand is a simple strand. Contact edges are orange.}
	\label{fig:bundle_example}
\end{minipage}
\hspace{20mm}
\begin{minipage}{.4\textwidth}
	\centering
    \includegraphics[width=0.6\linewidth]{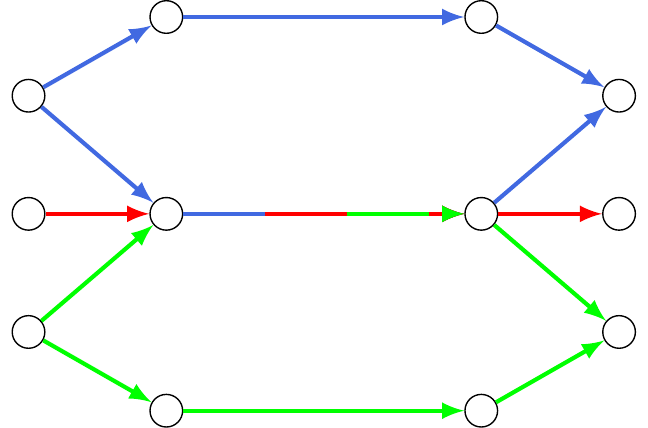}
	\captionof{figure}{Three bundles (blue, green, red indicate different times), with one strand from each merged together.}
	\label{fig:merge_example}
\end{minipage}
\end{figure}

Before formally giving the reduction, we illustrate a simple example of its construction.

\begin{example} \label{example:reduction}
Consider a toy $\LabelCover$ instance whose bipartite graph is a single edge, label set is $\Sigma = \{1,2\}$, and projection functions are shown:

\begin{figure}[H]
	\centering
    \includegraphics[width=0.25\textwidth]{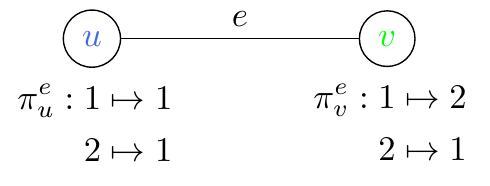}
	\label{fig:example_reduction_LC_instance}
\end{figure}

Our reduction outputs this corresponding 2-$\TSN$ instance:

\begin{figure}[H]
	\centering
    \includegraphics[width=0.4\textwidth]{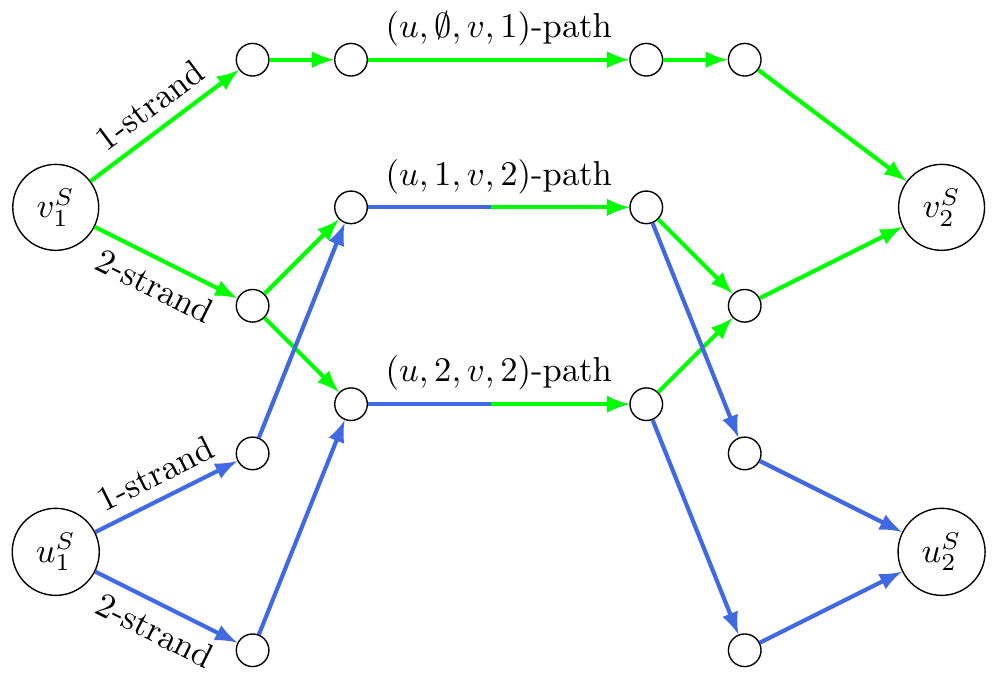}
	\label{fig:example_reduction_Steiner_instance}
\end{figure}

$G_1$ comprises the set of blue edges; $G_2$ is green. The demands are $(u_1^S, u_2^S, 1)$ and $(v_1^S, v_2^S, 2)$. The gadget for the $\LabelCover$ node $u$ (the blue subgraph) consists of two strands, one for each possible label. In the $v$-gadget (green subgraph), the strand corresponding to a labeling of `2' branches further, with one simple strand for each agreeing labeling of $u$. Finally, strands (more precisely, their contact edges) whose labels map to the same color are merged.

The input is a YES instance of $\LabelCover$ whose optimal labelings ($u$ gets either label 1 or 2, $v$ gets label 2) correspond to 2-$\TSN$ solutions of cost 1 (both gadgets traverse the $(u,1,v,2)$-path, or both traverse the $(u,2,v,2)$-path). If this were a NO instance and edge $e$ could not be satisfied, then the resulting 2-$\TSN$ gadgets would have no overlap.
\end{example}

%----------------------------------------------------------------------------------------
\subsection{Inapproximability for Two Demands}
We now formalize the reduction in the case of two demands; later, we extend this to any $k$.

\begin{theorem} \label{theorem:inapproximability_for_two_demands}
2-$\TSN$ and 2-$\DTSN$ are $\NP$-hard to approximate to within a factor of $2 - \epsilon$ for every constant $\epsilon > 0$. For 2-$\DTSN$, this holds even when the underlying graph is acyclic.
\end{theorem}

\begin{proof}
We describe a reduction from $\LabelCover$ to 2-$\DTSN$ with an acyclic graph. Given the $\LC$ instance $\left( G=(U,V,E), \Sigma, \Pi \right)$, construct a 2-$\DTSN$ instance ($\mathcal{G} = (G_1,G_2)$, along with two connectivity demands) as follows. Create nodes $u_1^S, \ldots, u_{|U|+1}^S$ and $v_1^S, \ldots, v_{|V|+1}^S$. Let there be a bundle from each $u_i^S$ to $u_{i+1}^S$; we call this the \emph{$u_i$-bundle}, since a choice of path from $u_i^S$ to $u_{i+1}^S$ in $\mathcal{G}$ will indicate a labeling of $u_i$ in $G$.

The $u_i$-bundle has a strand for each possible label $\ell \in \Sigma$. Each of these \emph{$\ell$-strands} consists of a chain of bundles---one for each edge $(u_i, v) \in E$. Finally, each such \emph{$(u_i,v)$-bundle} has a simple strand for each label $r \in \Sigma$ such that $\pi_{u_i}^{(u_i,v)}(\ell) = \pi_{v}^{(u_i,v)}(r)$; call this the \emph{$(u_i,\ell,v,r)$-path}. In other words, there is ultimately a simple strand for each possible labeling of $u_i$'s neighbor $v$ such that the two nodes are in agreement under their mutual edge constraint. If there are no such consistent labels $r$, then the $(u_i,v)$-bundle consists of just one simple strand, which is not associated with any $r$. Note that every minimal $u_1^S \to u_{|U|+1}^S$ path (that is, one that proceeds from one bundle to the next) has total weight exactly $|E|$.

Similarly, create a $v_j$-bundle from each $v_j^S$ to $v_{j+1}^S$, whose $r$-strands (for $r \in \Sigma$) are each a chain of bundles, one for each $(u,v_j) \in E$. Each $(u,v_j)$-bundle has a $(u,\ell,v_j,r)$-path for each agreeing labeling $\ell$ of the neighbor $u$, or a simple strand if there are no such labelings.

Set all the edges in the $u_i$-bundles to exist in $G_1$ only. Similarly the $v_j$-bundles exist solely in $G_2$. Now, for each $(u,\ell,v,r)$-path in $G_1$, merge it with the $(u,\ell,v,r)$-path in $G_2$, if it exists. The demands are $\mathcal{D} = \left\{ \left( u_1^S,u_{|U|+1}^S,1 \right), \left( v_1^S,v_{|V|+1}^S,2 \right) \right\}$.

We now analyze the reduction. The main idea is that any $u_i^S \to u_{i+1}^S$ path induces a labeling of $u_i$; thus the demand $\left( u_1^S,u_{|U|+1}^S,1 \right)$ ensures that any 2-$\DTSN$ solution indicates a labeling of all of $U$. Similarly, $\left( v_1^S,v_{|V|+1}^S,2 \right)$ forces an induced labeling of $V$. In the case of a YES instance of $\LabelCover$, these two connectivity demands can be satisfied by taking two paths with a large amount of overlap, resulting in a low-cost 2-$\DTSN$ solution. In contrast when we start with a NO instance of $\LabelCover$, any two paths we can choose to satisfy the 2-$\DTSN$ demands will be almost completely disjoint, resulting in a costly solution. We now fill in the details.

Suppose the $\LabelCover$ instance is a YES instance, so that there exists a labeling $\ell_{u}^*$ to each $u \in U$, and $r_{v}^*$ to each $v \in V$, such that for all edges $(u,v) \in E$, $\pi_u^{(u,v)}(\ell_u^*) = \pi_v^{(u,v)}(r_v^*)$. The following is an optimal solution $\mathcal{H}^*$ to the constructed 2-$\DTSN$ instance:
\begin{itemize}
	\item To satisfy the demand at time 1, for each $u$-bundle, take a path through the $\ell_u^*$-strand. In particular for each $(u,v)$-bundle in that strand, traverse the $(u,\ell_u^*,v,r_v^*)$-path.
    \item To satisfy the demand at time 2, for each $v$-bundle, take a path through the $r_v^*$-strand. In particular for each $(u,v)$-bundle in that strand, traverse the $(u,\ell_u^*,v,r_v^*)$-path.
\end{itemize}
In tallying the total edge cost, $\mathcal{H}^* \cap G_1$ (i.e. the subgraph at time 1) incurs a cost of $|E|$, since one contact edge in $\mathcal{G}$ is encountered for each edge in $G$. $\mathcal{H}^* \cap G_2$ accounts for no additional cost, since all contact edges correspond to a label which agrees with some neighbor's label, and hence were merged with the agreeing contact edge in $\mathcal{H}^* \cap G_1$. Clearly a solution of cost $|E|$ is the best possible, since every $u_1^S \to u_{|U|+1}^S$ path in $G_1$ (and every $v_1^S \to v_{|V|+1}^S$ path in $G_2$) contains at least $|E|$ contact edges.

Conversely suppose we started with a NO instance of $\LabelCover$, so that for any labeling $\ell_u^*$ to $u$ and $r_v^*$ to $v$, for at least $(1-\epsilon)|E|$ of the edges $(u,v) \in E$, we have $\pi_u^{(u,v)}(\ell_u^*) \neq \pi_v^{(u,v)}(r_v^*)$. By definition, any solution to the constructed 2-$\DTSN$ instance contains a simple $u_1^S \to u_{|U|+1}^S$ path $P_1 \in G_1$ and a simple $v_1^S \to v_{|V|+1}^S$ path $P_2 \in G_2$. $P_1$ alone incurs a cost of exactly $|E|$, since one contact edge in $\mathcal{G}$ is traversed for each edge in $G$. However, $P_1$ and $P_2$ share at most $\epsilon |E|$ contact edges (otherwise, by the merging process, this implies that more than $\epsilon |E|$ edges could be consistently labeled, which is a contradiction). Thus the solution has a total cost of at least $(2 - \epsilon)|E|$.

The directed temporal graph we constructed is acyclic, as every edge points ``to the right'' as in Example~\ref{example:reduction}. It follows from the gap between the YES and NO cases that 2-$\DTSN$ is $\NP$-hard to approximate to within a factor of $2 - \epsilon$ for every $\epsilon > 0$, even on DAGs. Finally, note that the same analysis holds for 2-$\TSN$, by simply making every edge undirected; however in this case the graph is clearly not acyclic.
\end{proof}

%----------------------------------------------------------------------------------------
\subsection{Inapproximability for General \textit{k}}

\inapproximabilityforgeneralk*
\begin{proof}
We perform a reduction from $\kPartiteHypergraphLabelCover$, a generalization of $\LabelCover$ to hypergraphs, to $\kTSN$, or $\kDTSN$ with an acyclic graph. Using the same ideas as in the $k=2$ case, we design $k$ demands composed of parallel paths corresponding to labelings, and merge edges so that a good global labeling corresponds to lots of overlaps between those paths. The full proof is left to \autoref{section:inapproximability_for_general_k}.
\end{proof}

Note that a $k$-approximation algorithm is to simply choose $\mathcal{H} = \bigcup_t \tilde{P_t}$, where $\tilde{P_t}$ is the shortest $a_t \to b_t$ path in $G_t$. Thus by \autoref{theorem:inapproximability_for_general_k}, essentially no better approximation is possible in terms of $k$ alone. In contrast, most classic Steiner problems have good approximation algorithms, or are even exactly solvable for constant $k$.
%----------------------------------------------------------------------------------------
%----------------------------------------------------------------------------------------

%----------------------------------------------------------------------------------------
%    MONOTONIC SPECIAL CASES
%----------------------------------------------------------------------------------------
\section{Monotonic Special Cases} \label{section:monotonic_special_cases}
In light of the tight lower bound, in this section we consider more tractable special cases of the temporal Steiner problems. Perhaps the simplest scenario in which an $o(k)$-approximation is possible is when $T \ll k$, in which case we can approximate individual static Steiner instances using the best known algorithms and combine them. A more interesting and quite natural restriction is that the changes over time are \emph{monotonic}:

\begin{definition}[$\Monotonic$ $\kTSN$ and $\kDTSN$]
In this special case of $\kTSN$ or $\kDTSN$, we have that for each $e \in E$ and $t \in [T]$, if $e \in G_t$, then $e \in G_{t'}$ for all $t' \geq t$.
\end{definition}

This is a natural extension of $\kDTSN$ based on our motivation from network biology. Oftentimes, once a protein (node) becomes active, it remains active for the remainder of the time span being considered. As the node and edge variants are equivalent (\autoref{proposition:variants_mutually_reducible}), these monotonic cases are good models for known biological observations.

Note also that this notion of monotonicity is analogous to the \emph{incremental} model often studied in online algorithms, in which graph elements can be added but not deleted (it is also analogous to the decremental model, since we can reverse the order of the frames). We now examine its effect on the complexity of the temporal Steiner problems.

%----------------------------------------------------------------------------------------
\subsection{Monotonicity in the Undirected Case}
In the undirected case, monotonicity has a simple effect: it makes $\kTSN$ equivalent to the following well-studied problem:

\begin{definition}[$\PrioritySteinerTree$ \cite{charikar2004resource}]
The input is a weighted undirected multigraph $G=(V,E,w)$, a \emph{priority level} $p(e)$ for each $e \in E$, and a set of $k$ demands $(a_i,b_i)$, each with priority $p(a_i,b_i)$. The output is a minimum-weight forest $F \subseteq G$ that contains, between each $a_i$ and $b_i$, a path in which every edge $e$ has priority $p(e) \leq p(a_i,b_i)$.
\end{definition}

\begin{lemma} \label{lemma:monotonic_kTSN_and_priority_steiner_tree_are_equivalent}
$\PrioritySteinerTree$ and $\Monotonic$ $\kTSN$ have the same approximability.
\end{lemma}

\begin{proof}
We transform an instance of $\PrioritySteinerTree$ into an instance of $\Monotonic$ $\kTSN$ as follows: the set of priorities becomes the set of times; if an edge $e$ has priority $p(e)$, it now exists at all times $t \geq p(e)$; if a demand $(a_i,b_i)$ has priority $p(a_i,b_i)$, it now becomes $(a_i,b_i,p(a_i,b_i))$. If there are parallel multiedges, break up each such edge into two edges of half the original weight, joined by a new node. Given a solution $H \subseteq G$ to this $\kTSN$ instance, contracting any edges that were originally multiedges gives a $\PrioritySteinerTree$ solution of the same cost. This reduction also works in the opposite direction (in this case there are no multiedges), which shows the equivalence.
\end{proof}

$\PrioritySteinerTree$ is known to have a $O(\log k)$-approximation algorithm \cite{charikar2004resource}, as well as a lower bound of $\Omega(\log \log n)$ assuming $\NP \notin \DTIME(n^{\log \log \log n})$ \cite{chuzhoy2008approximability}. This combined with Lemma~\ref{lemma:monotonic_kTSN_and_priority_steiner_tree_are_equivalent} proves \autoref{theorem:monotonic_kTSN_approximability}.

%----------------------------------------------------------------------------------------
\subsection{Monotonicity in the Directed Case}
Now we consider the directed case, and in particular the special case in which all demands originate from a common source:

\begin{definition}[$\SingleSource$ $\Monotonic$ $\kDTSN$]
This problem is a special case of $\Monotonic$ $\kDTSN$ in which the demands are precisely $(a,b_1,t_1), (a,b_2,t_2), \ldots, (a,b_k,t_k)$, for some \emph{root} $a \in V$. We can assume w.l.o.g. that $t_1 \leq t_2 \leq \cdots \leq t_k$.
\end{definition}

Our goal in this subsection is to show that in terms of approximability, this problem is equivalent to $\DirectedSteinerTree$ ($\DST$), and hence the known bounds for $\DST$ apply. Note that one side of the equivalence is immediate, as $\SingleSource$ $\Monotonic$ $\kDTSN$ is a generalization of $\DST$ and therefore no easier to approximate.

\begin{lemma} \label{lemma:monotonic_singlesource_kDTSN_and_DST_are_equivalent}
$\Monotonic$ $\SingleSource$ $\kDTSN$ and $\DirectedSteinerTree$ have the same approximability.
\end{lemma}

For the remainder of this section, we refer to $\Monotonic$ $\SingleSource$ $\kDTSN$ as simply $\kDTSN$. To prove Lemma~\ref{lemma:monotonic_singlesource_kDTSN_and_DST_are_equivalent}, it remains to give an approximation-preserving reduction from $\kDTSN$ to $\DST$.

\subparagraph*{The reduction}
Given a $\kDTSN$ instance $(G_1=(V,E_1), G_2=(V,E_2), \ldots, G_T=(V,E_T), \mathcal{D})$ with underlying graph $\mathcal{G} = (V,E)$, we construct a $\DST$ instance $(G'=(V',E') ,D')$ as follows:
\begin{itemize}
    \item $G'$ contains a vertex $v^i$ for each $v \in V$ and each $i \in [k]$. It contains an edge $(u^i,v^i)$ with weight $w(u,v)$ for each $(u,v) \in E_i$. Additionally, it contains a zero-weight edge $(v^i,v^{i+1})$ for each $v \in V$ and each $i \in [k]$.
    \item $D'$ contains a demand $(a^1, b_i^{t_i})$ for each $(a, b_i, t_i) \in \mathcal{D}$.
\end{itemize}
Now consider the $\DST$ instance $(G',D')$.

\begin{lemma} \label{lemma:kDTSN_solution_implies_DST_solution}
If the $\kDTSN$ instance $(G_1,\ldots,G_T, \mathcal{D})$ has a solution of cost $C$, then the constructed $\DST$ instance $(G',D')$ has a solution of cost at most $C$.
\end{lemma}

\begin{proof}
Let $\mathcal{H} \subseteq \mathcal{G}$ be a $\kDTSN$ solution having cost $C$. For any edge $(u,v) \in E(\mathcal{H})$, define the \emph{earliest necessary time} of $(u,v)$ to be the minimum $t_i$ such that removing $(u,v)$ would cause $\mathcal{H}$ not to satisfy demand $(a,b_i,t_i)$.

\begin{claim} \label{claim:directed_tree_solution}
There exists a solution $\mathcal{T} \subseteq \mathcal{H}$ that is a directed tree and has cost at most $C$. Moreover for every path $P_i$ in $\mathcal{T}$ from the root $a$ to some target $b_i$, as we traverse $P_i$ from $a$ to $b_i$, the earliest necessary times of the edges are non-decreasing.
\end{claim}

\begin{proof}[Proof of \autoref{claim:directed_tree_solution}]
Consider a partition of $\mathcal{H}$ into edge-disjoint subgraphs $\mathcal{H}_1,\ldots,\mathcal{H}_k$, where $\mathcal{H}_i$ is the subgraph whose edges have earliest necessary time $t_i$. Clearly each $\mathcal{H}_i$ is a single component.

If there is a directed cycle or parallel paths in the first subgraph $\mathcal{H}_1$, then there is an edge $e \in E(\mathcal{H}_1)$ whose removal does not cause $\mathcal{H}_1$ to satisfy fewer demands at time $t_1$. Moreover by monotonicity, removing $e$ also does not cause $\mathcal{H}$ to satisfy fewer demands at any future times. Hence there exists a directed tree $\mathcal{T}_1 \subseteq \mathcal{H}_1$ such that $\mathcal{T}_1 \cup \left(\bigcup_{i=2}^k \mathcal{H}_i\right)$ has cost at most $C$ and still satisfies $\mathcal{D}$.

Now suppose by induction that for some $j \in [k-1]$, $\bigcup_{i=1}^j \mathcal{T}_i$ is a tree such that $\left(\bigcup_{i=1}^j \mathcal{T}_i\right) \cup \left(\bigcup_{i=j+1}^k \mathcal{H}_i\right)$ has cost at most $C$ and satisfies $\mathcal{D}$. Consider the partial solution $\left(\bigcup_{i=1}^j \mathcal{T}_i\right) \cup \mathcal{H}_{j+1}$; if this subgraph is not a directed tree, then there must be an edge $(u,v) \in E(\mathcal{H}_{j+1})$ such that $v$ has another in-edge in the subgraph. However by monotonicity, $(u,v)$ does not help satisfy any new demands, as $v$ is already reached by some other path from the root. Hence by removing all such redundant edges, we have $\mathcal{T}_{j+1} \subseteq \mathcal{H}_{j+1}$ such that $\left(\bigcup_{i=1}^{j+1} \mathcal{T}_i\right) \cup \left(\bigcup_{i=j+2}^k \mathcal{H}_i\right)$ has cost at most $C$ and satisfies $\mathcal{D}$, which completes the inductive step.

We conclude that $\mathcal{T} := \bigcup_{i=1}^k \mathcal{T}_i \subseteq \mathcal{H}$ is a tree of cost at most $C$ satisfying $\mathcal{D}$. Observe also that by construction, $\mathcal{T}$ has the property that if we traverse any $a \to b_i$ path, the earliest necessary times of the edges never decrease.
\end{proof}

Now let $\mathcal{T}$ be the $\kDTSN$ solution guaranteed to exist by \autoref{claim:directed_tree_solution}. Consider the subgraph $H' \subseteq G'$ formed by adding, for each $(u,v) \in E(\mathcal{T})$, the edge $(u^t,v^t) \in E'$ where $t$ is the earliest necessary time of $(u,v)$ in $E(\mathcal{H})$. In addition, add all the free edges $(v^i,v^{i+1})$. Since $w(u^t,v^t) = w(u,v)$ by construction, $\text{cost}(H') \leq \text{cost}(\mathcal{T}) \leq C$.

To see that $H'$ is a valid solution, consider any demand $(a^1,b_i^{t_i})$. Recall that $\mathcal{T}$ has a unique $a \to b_i$ path $P_i$ along which the earliest necessary times are nondecreasing. We added to $H'$ each of these edges at the level corresponding to its earliest necessary time; moreover, whenever there are adjacent edges $(u,v), (v,x) \in P_i$ with earliest necessary times $t$ and $t' \geq t$ respectively, there exist in $H'$ free edges $(v^t,v^{t+1}),\ldots,(v^{t'-1},v^{t'})$. Thus $H'$ contains an $a^1 \to b_i^{t_i}$ path, which completes the proof.
\end{proof}

\begin{lemma} \label{lemma:DST_solution_implies_kDTSN_solution}
If the constructed $\DST$ instance $(G',D')$ has a solution of cost $C$, then the original $\kDTSN$ instance $(G_1,\ldots,G_T,\mathcal{D})$ has a solution of cost at most $C$.
\end{lemma}

\begin{proof}
First note that any $\DST$ solution ought to be a tree; let $T' \subseteq G'$ be such a solution of cost $C$. For each $(u,v) \in G$, $T'$ might as well use at most one edge of the form $(u^i,v^i)$, since if it uses more, it can be improved by using only the one with minimum $i$, then taking the free edges $(v^i,v^{i+1})$ as needed. We create a $\kDTSN$ solution $\mathcal{T} \subseteq \mathcal{G}$ as follows: for each $(u^i,v^i) \in E(T')$, add $(u,v)$ to $\mathcal{T}$. Since $w(u,v) = w(u^i,v^i)$ by design, we have $\text{cost}(\mathcal{T}) \leq \text{cost}(T') \leq C$. Finally, since each $a^1 \to b_i^{t_i}$ path in $G'$ has a corresponding path in $\mathcal{G}$ by construction, $\mathcal{T}$ satisfies all the demands.
\end{proof}

Lemma~\ref{lemma:monotonic_singlesource_kDTSN_and_DST_are_equivalent} follows from Lemma~\ref{lemma:kDTSN_solution_implies_DST_solution} and Lemma~\ref{lemma:DST_solution_implies_kDTSN_solution}. Finally we get the main result of this subsection:

\monotonicsinglesourcekDTSNapproximability*

\begin{proof}
This follows by composing the reduction with the algorithm of Charikar et al. \cite{charikar1999approximation} for $\DirectedSteinerTree$, which achieves ratio $O(k^\epsilon)$ for every $\epsilon > 0$. More precisely they give a $i^2(i-1) k^{1/i}$-approximation for any integer $i \geq 1$, in time $O(n^i k^{2i})$. The lower bound follows from the known hardness of $\DST$ due to Halperin and Krauthgamer~\cite{halperin2003polylogarithmic}.
\end{proof}

In \autoref{section:explicit_algorithm}, we show how to modify the algorithm of Charikar et al. to arrive at a simple, explicit algorithm for $\Monotonic$ $\SingleSource$ $\kDTSN$ achieving the same upper bound as \autoref{theorem:monotonic_singlesource_kDTSN_approximability}.
%----------------------------------------------------------------------------------------
%----------------------------------------------------------------------------------------

%----------------------------------------------------------------------------------------
%    DISCUSSION
%----------------------------------------------------------------------------------------
\section{Discussion}
Our upper bound on the \emph{single-source} monotonic directed case leaves the following open:

\begin{openquestion}
Is there a nontrivial approximation algorithm for the monotonic directed problem with arbitrary demands?
\end{openquestion}

It would be particularly elegant if this could be resolved by reducing the problem to $\DirectedSteinerNetwork$, just as the single-source case was reduced to $\DirectedSteinerTree$. However this seems to require new techniques, as a simple counterexample shows that \autoref{claim:directed_tree_solution} does not hold for multi-source demands.

Secondly, our reduction from $\kPartiteHypergraphLabelCover$ to the temporal Steiner problems depends, in more than one way, on $k$ being fixed. One might hope, then, to obtain a nontrivial approximation in terms of $n$ and $T$ instead of $k$. When $k = \Theta(n^2 T)$, the shortest paths heuristic only guarantees a ratio of $O(n^2 T)$. Using an algorithm of \cite{feldman2009improved} that approximates $\DSN$ to within $O(n^{4/5 + \epsilon})$, we can improve this to $O(n^{4/5 + \epsilon} \cdot T)$.

On the complexity side, the current best approximation lower bound for $\kDTSN$ purely in terms of $n$ is the same as that for $\DSN$: $\Omega(2^{\log^{1-\epsilon} n})$ unless $\NP$ has quasipolynomial-time algorithms \cite{dodis1999design}. Recall a conceptual message of \autoref{theorem:inapproximability_for_general_k}: if we seek an approximation in terms of $k$, then the best we can do is to consider the graph at each individual time (in the reduction, each demand occurred at a separate time). We ask if there is a strong lower bound that demonstrates a similar message when trying to approximate in terms of $n$ and $T$:

\begin{openquestion}
For $\kTSN$ or $\kDTSN$, is there a better approximation algorithm in terms of $n$ and/or $T$? Can we show an inapproximability of $\Omega(T \cdot 2^{\log^{1-\epsilon} n})$?
\end{openquestion}
%----------------------------------------------------------------------------------------
%----------------------------------------------------------------------------------------

%%
%% Bibliography
%%

%% Either use bibtex (recommended), 

\bibliography{main}

%% .. or use the thebibliography environment explicitly

\appendix

%----------------------------------------------------------------------------------------
%    PROBLEM VARIANTS
%----------------------------------------------------------------------------------------
\section{Problem Variants} \label{section:problem_variants}

There are several natural ways to formulate the temporal Steiner network problem, depending on whether the edges are changing over time, or the nodes, or both.

\begin{definition}[$\kTemporalSteinerNetwork$ (edge variant)]
This is the formulation described in the Introduction: the inputs are $G_1=(V,E_1),\ldots,G_T=(V,E_T)$, $w(\cdot)$, and $\mathcal{D}=\{(a_i,b_i,t_i)\}$. The task is to find a minimum-weight subgraph $\mathcal{H} \subseteq \mathcal{G}$ that satisfies all of the demands.
\end{definition}

\begin{definition}[$\kTemporalSteinerNetwork$ (node variant)]
Let the underlying graph be $\mathcal{G} = (V,E)$. The inputs are $G_1=(V_1,E(V_1)),\ldots,G_T=(V_T,E(V_T))$, $w(\cdot)$, and $\mathcal{D}$. Here, $E(V_t) \subseteq E$ denotes the edges induced by $V_t \subseteq V$. A path satisfies a demand at time $t$ iff all edges along that path exist in $G_t$.
\end{definition}

\begin{definition}[$\kTemporalSteinerNetwork$ (node and edge variant)]
The inputs are $G_1=(V_1,E_1),\ldots,G_T=(V_T,E_T)$, $w(\cdot)$, and $\mathcal{D}$. This is the same as the node variant except that each $E_t$ can be any subset of $E(V_t)$.
\end{definition}

Similarly, define the corresponding directed problem $\kDirectedTemporalSteinerNetwork$ ($\kDTSN$) with the same three variants. The only difference is that the edges are directed, and a demand $(a,b,t)$ must be satisfied by a directed $a \to b$ path in $G_t$.

The following observation enables all our results to apply to all problem variants.

\variantsmutuallyreducible*
\begin{proof}
The following statements shall hold for both undirected and directed versions. Clearly the node-and-edge variant generalizes the other two. It suffices to show two more directions:
\begin{itemize}
	\item (Node-and-edge reduces to node) Let $(u,v)$ be an edge existent at a set of times $\tau(u,v)$, whose endpoints exist at times $\tau(u)$ and $\tau(v)$. To make this a node-temporal instance, create an intermediate node $x_{(u,v)}$ existent at times $\tau(u,v)$, an edge $(u,x_{(u,v)})$ with the original weight $w(u,v)$, and an edge $(x_{(u,v)},v)$ with zero weight. A solution of cost $W$ in the node-and-edge instance corresponds to a node-temporal solution of cost $W$, and vice-versa.
    \item (Node reduces to edge) Let $(u,v)$ be an edge whose endpoints exist at times $\tau(u)$ and $\tau(v)$. To make this an edge-temporal instance, let $(u,v)$ exist at times $\tau(u,v) := \tau(u) \cap \tau(v)$. Let every node exist at all times; let the edges retain their original weights. A solution of cost $W$ in the node-temporal instance corresponds to an edge-temporal solution of cost $W$, and vice-versa.
\end{itemize}
\end{proof}
%----------------------------------------------------------------------------------------
%----------------------------------------------------------------------------------------

%----------------------------------------------------------------------------------------
%    INAPPROXIMABILITY FOR GENERAL K
%----------------------------------------------------------------------------------------
\section{Inapproximability for General \textit{k}} \label{section:inapproximability_for_general_k}

Here we prove our main theorem, showing optimal hardness for any number of demands. To do this, we introduce a generalization of $\LabelCover$ to partite hypergraphs:

\begin{definition}[$\kPartiteHypergraphLabelCover$ ($\kPHLC$)]
An instance of this problem consists of a $k$-partite, $k$-regular hypergraph $G = (V_1,\ldots,V_k,E)$ (that is, each edge contains exactly one vertex from each of the $k$ parts) and a set of possible labels $\Sigma$. The input also includes, for each hyperedge $e \in E$, a \emph{projection function} $\pi_v^{e} : \Sigma \to C$ for each $v \in e$; $\Pi$ is the set of all such functions. A labeling of $G$ is a function $\phi : \bigcup_{i=1}^k V_i \to \Sigma$ assigning each node a label. There are two notions of edge satisfaction under a labeling $\phi$:
\begin{itemize}
	\item $\phi$ \emph{strongly satisfies} a hyperedge $e = (v_1,\ldots,v_k)$ iff the labels of all its vertices are mapped to the same color, i.e. $\pi_{v_i}^e(\phi(v_i)) = \pi_{v_j}^e(\phi(v_j))$ for all $i,j \in [k]$.
    \item $\phi$ \emph{weakly satisfies} a hyperedge $e = (v_1,\ldots,v_k)$ iff there exists some pair of vertices $v_i$, $v_j$ whose labels are mapped to the same color, i.e. $\pi_{v_i}^e(\phi(v_i)) = \pi_{v_j}^e(\phi(v_j))$ for some $i \neq j \in [k]$.
\end{itemize}
\end{definition}

\begin{theorem}\label{theorem:hypergraph_label_cover_hardness}
For every $\epsilon > 0$ and every fixed integer $k \geq 2$, there is a constant $|\Sigma|$ such that the following promise problem is $\NP$-hard: Given a $\kPartiteHypergraphLabelCover$ instance $(G, \Sigma, \Pi)$, distinguish between the following cases:
\begin{itemize}
	\item (YES instance) There exists a labeling of $G$ that strongly satisfies every edge.
    \item (NO instance) Every labeling of $G$ weakly satisfies at most $\epsilon |E|$ edges.
\end{itemize}
\end{theorem}

\autoref{theorem:hypergraph_label_cover_hardness} follows from Feige's $k$-prover system \cite{feige1998threshold} by taking the number of repetitions to be (a constant depending on $k$ and $\epsilon$) large enough so that the error probability drops below $\epsilon$.

The proof of $(k-\epsilon)$-hardness follows the same outline as the $k=2$ case (\autoref{theorem:inapproximability_for_two_demands}).

\inapproximabilityforgeneralk*
\begin{proof}
Given the $\kPHLC$ instance $(G=(V_1,\ldots,V_k,E), \Sigma, \Pi)$, and letting $v_{t,i}$ denote the $i$-th node in $V_t$, construct a $\kDTSN$ instance ($\mathcal{G} = (G_1,\ldots,G_k)$, along with $k$ demands) as follows. For every $t \in [k]$, create nodes $v_{t,1}^S, \ldots, v_{t,|V_t|+1}^S$. Create a \emph{$v_{t,i}$-bundle} from each $v_{t,i}^S$ to $v_{t,i+1}^S$, whose $\ell$-strands (for $\ell \in \Sigma$) are each a chain of bundles, one for each incident hyperedge $e = (v_{1,i_1},\ldots,v_{t,i} ,\ldots,v_{k,i_k}) \in E$. Each $(v_{1,i_1},\ldots,v_{t,i} ,\ldots,v_{k,i_k})$-bundle has a $(v_{1,i_1},\ell_1,\ldots,v_{t,i},\ell_t,\ldots,v_{k,i_k},\ell_k)$-path for each agreeing combination of labels---that is, every $k$-tuple $(\ell_1,\ldots,\ell_t,\ldots,\ell_k)$ such that $\pi_{v_{1,i_1}}^e(\ell_1) = \cdots = \pi_{v_{t,i}}^e(\ell_t) = \cdots = \pi_{v_{k,i_k}}^e(\ell_k)$, where $e$ is the shared edge. If there are no such combinations, then the $e$-bundle is a single simple strand.

For each $t \in [k]$, set all the edges in the $v_{t,i}$-bundles to exist in $G_t$ only. Now, for each $(v_{1,i_1},\ell_1,\ldots,v_{k,i_k},\ell_k)$, merge together the $(v_{1,i_1},\ell_1,\ldots,v_{k,i_k},\ell_k)$-paths across all $G_t$ that have such a strand. Finally, the connectivity demands are $\mathcal{D} = \left\{ \left( v_{t,1}^S,v_{t,|V_t|+1}^S,t \right) : t \in [k] \right\}$.

The analysis follows the $k=2$ case. Suppose we have a YES instance of $\kPHLC$, with optimal labeling $\ell_v^*$ to each node $v \in \bigcup_{t=1}^k V_t$. Then an optimal solution $\mathcal{H}^*$ to the constructed $\kDTSN$ instance is to traverse, at each time $t$ and for each $v_{t,i}$-bundle, the path through the $\ell_{v_{t,i}}^*$-strand. In particular for each $(v_{1,i_1},\ldots,v_{k,i_k})$-bundle in that strand, traverse the $(v_{1,i_1},\ell_1^*,\ldots,v_{k,i_k},\ell_k^*)$-path.

In tallying the total edge cost, $\mathcal{H}^* \cap G_1$ (the subgraph at time 1) incurs a cost of $|E|$, one for each contact edge. The subgraphs of $\mathcal{H}^*$ at times $2,\ldots,k$ account for no additional cost, since all contact edges correspond to a label which agrees with all its neighbors' labels, and hence were merged with the agreeing contact edges in the other subgraphs.

Conversely suppose we have a NO instance of $\kPHLC$, so that for any labeling $\ell_v^*$, for at least $(1-\epsilon)|E|$ hyperedges $e$, the projection functions of all nodes in $e$ disagree. By definition, any solution to the constructed $\kDTSN$ instance contains a simple $v_{t,1}^S \to v_{t,|V_t|+1}^S$ path $P_t$ at each time $t$. As before, $P_1$ alone incurs a cost of exactly $|E|$. However, at least $(1-\epsilon)|E|$ of the hyperedges in $G$ cannot be weakly satisfied; for these hyperedges $e$, for every pair of neighbors $v_{t,i_t}, v_{t',i_{t'}} \in e$, there is no path through the $e$-bundle in $v_{t,i_t}$'s $\ell_{v_{t,i_t}}^*$-strand that is merged with any of the paths through the $e$-bundle in $v_{t',i_{t'}}$'s $\ell_{v_{t,i_{t'}}}^*$-strand (for otherwise, it would indicate a labeling that weakly satisfies $e$ in the $\kPHLC$ instance). Therefore paths $P_2, \ldots, P_k$ each contribute at least $(1-\epsilon)|E|$ additional cost, so the solution has total cost at least $(1-\epsilon)|E| \cdot k$.

The directed temporal graph we constructed is acyclic. It follows from the gap between the YES and NO cases that $\kDTSN$ is $\NP$-hard to approximate to within a factor of $k - \epsilon$ for every constant $\epsilon > 0$, even on DAGs. As before, the same analysis holds for the undirected problem $\kTSN$ by undirecting the edges.
\end{proof}
%----------------------------------------------------------------------------------------
%----------------------------------------------------------------------------------------

%----------------------------------------------------------------------------------------
%    AN EXACT ALGORITHM
%----------------------------------------------------------------------------------------
\section{An Exact Algorithm} \label{section:exact_algorithm}

We can derive a natural integer linear program for $\kDTSN$ by first reducing it to a simpler special case, then writing the ILP for this simpler version. The simplified case we consider is one in which all demands share a source and target, and each demand occupies a distinct time.

\begin{definition}[$\Simple$ $\kDTSN$ (node variant)]
This is the special case of $\kDTSN$ in which the demands are precisely $(a,b,1), (a,b,2), \ldots, (a,b,k)$, for some common $a,b \in V$.
\end{definition}
The following reduction demonstrates that demands can always be simplified at the expense of more time points.

\begin{restatable}[$\kDTSN$ reduces to $\Simple$ $\kDTSN$ (node variant)]{lemma}{reductiontosimple}
\label{lemma:reduction_to_simple}
An instance of $\kDTSN$ with underlying graph $G=(V,E)$ and $k$ demands (not necessarily at distinct times) can be converted to an instance of $\Simple$ $\kDTSN$ with $2k+|V|$ nodes, $2k+|E|$ edges, and $k$ demands at distinct times.
\end{restatable}

\begin{proof}
Suppose we are given an instance of $\kDTSN$ with temporal graph $\mathcal{G}=(G_1,\ldots,G_T)$, underlying graph $G=(V,E)$, and demands $\mathcal{D}=\{(a_i,b_i,t_i) : i \in [k]\}$. By \autoref{proposition:variants_mutually_reducible}, we can assume this is a node-temporal instance. We build a new instance with temporal graph $\mathcal{G}'=(G_1',\ldots,G_T')$, underlying graph $G'=(V',E')$, and demands $\mathcal{D}'$.

Initialize $\mathcal{G}'$ to $\mathcal{G}$. Define a new set of times $[k]$. Add to $\mathcal{G}'$ the new nodes $a$ and $b$, which exist at all times/in all frames $G_i$. For all $v \in V$ and $i \in [k]$, if $v \in G_{t_i}$, then let $v$ exist in $G_i'$ as well. For each $(a_i,b_i,t_i) \in \mathcal{D}$,
\begin{enumerate}
	\item Create new nodes $x_i$, $y_i$. Create zero-weight edges $(a,x_i)$, $(x_i,a_i)$, $(b_i,y_i)$, and $(y_i,b)$.
    \item Let $x_i$ and $y_i$ exist only in frame $G_i'$.
\end{enumerate}
Lastly, the demands are $\mathcal{D}' = \{(a,b,i) : i \in [k]\}$.

Given a solution $\mathcal{H}' \subseteq \mathcal{G}'$ containing an $a \to b$ path at every time $i \in [k]$, we can simply exclude nodes $a$, $b$, $\{x_i\}$, and $\{y_i\}$ to obtain a solution $\mathcal{H} \subseteq \mathcal{G}$ to the original instance, which contains an $a_i \to b_i$ path in $G_{t_i}$ for all $i \in [k]$, and has the same cost. The converse is also true by including these nodes.
\end{proof}

The node variant of $\Simple$ $\kDTSN$ has a natural integer programming formulation in terms of flows:

\begin{align}
& \text{minimize}	& \sum_{(u,v) \in E} d_{uv} \cdot w(u,v) \\
& \text{subject to}	& d_{uv} &\geq d_{uvt}
						&& \forall t \in [T], \,\, (u,v) \in E_t
                        \label{edge_decision_constraint} \\
&					& \sum_{(u,v) \in E_t} d_{uvt} - \sum_{(v,w) \in E_t} d_{vwt} &= 0
						&& \forall t \in [T], \,\, v \notin \{a,b\}
                        \label{flow_conservation_constraint} \\
&					& \sum_{(a,u) \in E_t} d_{aut} &= 1
						&& \forall t \in [T]
                        \label{source_flow_constraint} \\
&					&  \sum_{(u,b) \in E_t} d_{ubt} &= 1
						&& \forall t \in [T]
                        \label{target_flow_constraint} \\
&					&  d_{uvt} &\in \{0,1\}
						&& \forall t \in [T], \,\, (u,v) \in E_t
                        \label{variable_integrality_constraint}
\end{align}

Each variable $d_{uvt}$ denotes the flow through edge $(u,v)$ at time $t$, if it exists. Constraint (\ref{edge_decision_constraint}) ensures that if an edge is used at any time, it is chosen as part of the solution subgraph. (\ref{flow_conservation_constraint}) enforces flow conservation at all nodes and all times. (\ref{source_flow_constraint}) and (\ref{target_flow_constraint}) impose non-zero flow from $a$ to $b$ at all times.

We implemented the reduction of \autoref{lemma:reduction_to_simple} and the above ILP in Python, using Gurobi optimization software (code is available at: \url{https://github.com/YosefLab/dynamic_connectivity}).
%----------------------------------------------------------------------------------------
%----------------------------------------------------------------------------------------

%----------------------------------------------------------------------------------------
%    EXPLICIT ALGORITHM FOR MONOTONIC SINGLE-SOURCE KDTSN
%----------------------------------------------------------------------------------------
\section{Explicit Algorithm for Monotonic Single-Source \textit{k}-DTSN} \label{section:explicit_algorithm}
We provide a modified version of the approximation algorithm presented in Charikar et al. \cite{charikar1999approximation} for $\DirectedSteinerTree$ ($\DST$), which achieves the same approximation ratio for our problem $\Monotonic$ $\SingleSource$ $\kDTSN$.

We provide a similar explanation as of that presented in Charikar et al. Consider a trivial approximation algorithm, where we take the shortest path from the source to each individual target. Consider the example where there are edges of cost $C-\epsilon$ to each target, and a vertex $v$ with distance $C$ from the source, and with distance $0$ to each target. In such a case, this trivial approximation algorithm will have an $O(k)$ approximation. Consider instead, an algorithm which found from the root, an intermediary vertex $v$, which was connected to all the targets via shortest path. In the case of the above example, this would find us the optimal subgraph. The algorithm below generalizes this process, by progressively finding optimal substructures with good cost relative to the number of targets connected. We show that this algorithm provides a good approximation ratio.

\begin{definition}[Metric closure of a temporal graph]
For a directed temporal graph $\mathcal{G} = (G_1=(V,E_1), G_2=(V,E_2), \ldots, G_T=(V,E_T))$, define its \emph{metric closure} to be $\tilde{G} = (V,E,\tilde{w})$ where $E = \bigcup_t E_t$ and $\tilde{w}(u,v,t)$ is the length of the shortest $u \to v$ path in $G_t$ (note that in contrast with $w$, $\tilde{w}$ takes three arguments).
\end{definition}

\begin{definition}[$V(T)$]
Let $T$ be a tree with root $r$. We say a demand of the form $(r,b,t)$ is \emph{satisfied} by $T$ if there is a path in $T$ from $r$ to $b$ at time $t$. $V(T)$ is then the set of demands satisfied by $T$.
\end{definition}

\begin{definition}[$D(T)$]
The \emph{density} of a tree $T$ is $D(T) = cost(T)/|V(T)|$, where $cost(T)$ is the sum of edge weights of $T$.
\end{definition}

\begin{algorithm}[H]
\begin{algorithmic}[1]
	\Statex
	\Function{$A_i$}{transitive closure $G = (V,E,w)$, $r$, $t$, $k$, $\mathcal{D} \subseteq V \times [T]$}
    	\If{ $(r, b_i, t_i)$ does not exist for least $k (b_i, t_i) \in \mathcal{D}, t_i \geq t$ } \Return \textsc{no solution}
        \EndIf
        
        \State $T \gets \emptyset$
        
        \While{$k > 0$}
        	\State $T_{best} \gets \emptyset$
            
            \ForAll{$(v,t') \in V \times [T], t' \geq t$ and $k', 1 \geq k' \geq k$}
                \State $T' \gets \A_{i-1}(G,v,t', k', \mathcal{D}) \cup {(r,v, t')}$
                \If{ $d(T_{BEST})> d(T')$} $T_{BEST} \gets T'$ \Comment{\textcolor{red}{Demand $i$ satisfied only if edge to $b_i$ at $t_i$ (ie $(x,b_i,t_i)$ for some x)}}
                \EndIf
            \EndFor
                \State $T \gets T \cup T_{BEST}$; $k \gets k - |D \cap V(T_{BEST})|$; $X \gets X - V(T_{BEST})$
            \EndWhile
        
        \State \Return $T$
    \EndFunction
\end{algorithmic}
\end{algorithm}

The way we will prove the approximation ratio of this algorithm is to show that it behaves precisely as the algorithm of Charikar et al. does, when given as input the $\DST$ instance produced by our reduction from $\Monotonic$ $\SingleSource$ $\kDTSN$ (Lemma~\ref{lemma:monotonic_singlesource_kDTSN_and_DST_are_equivalent}).

\begin{proposition}
The algorithm above is equivalent to the algorithm of Charikar et al., when applied to the $\DST$ instance output by the reduction of Lemma~\ref{lemma:monotonic_singlesource_kDTSN_and_DST_are_equivalent}.
\end{proposition}
\begin{proof}
To see this, note that in our reduced instance, we see a collection of vertices, ${v^1	,...,v^{|T|}}$. Therefore, the only equivalent modification's needed to the original algorithm are:
\begin{itemize}
	\item In the input, rather than keeping track of the current root as some vertex $v^i$, keep track of $v$ at the current timepoint instead, i.e. $(v,i)$. 
	\item The distance from some $v^i$ to $x^j, j \geq i$ is simply the distance from $v$ to $x$ at time $j$, i.e. $\tilde{w}(v,x,j)$.
	\item Instead of looping through all vertices in the form ${v^1,\ldots,v^{|T|}}$, we instead loop through all vertices, and all time points.
\end{itemize}
Therefore this algorithm guarantees the same approximation ratio for $\Monotonic$ $\SingleSource$ $\kDTSN$ as the original algorithm achieved for $\DST$. In particular for all $i>1$, $A_i(G, a, 0, k, D)$ provides an $i^2(i-1)k^{1/i}$ approximation to $\kDTSN$, in time $O(n^i k^{2i})$ \cite{charikar1999approximation,helvig2001improved}\footnote{The first paper \cite{charikar1999approximation} incorrectly claims a bound of $i(i-1)k^{1/i}$; this was corrected in \cite{helvig2001improved}.}.
\end{proof}
%----------------------------------------------------------------------------------------
%----------------------------------------------------------------------------------------

%----------------------------------------------------------------------------------------
%    APPLICATIONS TO COMPUTATIONAL BIOLOGY
%----------------------------------------------------------------------------------------
\section{Applications to Computational Biology} \label{section:computational_biology}

In molecular biology applications, networks are routinely defined over a wide range of basic entities such as proteins, genes, metabolites, or drugs, which serve as nodes. The edges in these networks can have different meaning, depending on the particular context. For instance, in protein-protein interaction (PPI) networks, edges represent physical contact between proteins, either within stable multi-subunit complexes or through transient causal interactions (i.e., an edge $(x,y)$ means that protein $x$ can cause a change to the molecular structure of protein $y$ and thereby alter its activity). The body of knowledge encapsulated within the human PPI network (tens of thousands of nodes and hundreds of thousands of edges in current databases, curated from thousands of studies \cite{Chatr15}) is routinely used by computational biologists to generate hypotheses of how various signals are transduced in eukaryotic cells \cite{Wu2013tcell}. The basic premise is that a process that starts with a change to the activity of protein $u$ and ends with the activity of protein $v$ must be propagated through a chain of interactions between $u$ and $v$. The natural extension regards a process with a certain collection of protein pairs $\{(u_1,v_1),\ldots,(u_k,v_k)\}$, where we are looking for a chain of interactions between each $u_i$ and $v_i$. In most applications, the identity of $u_i$ and $v_i$ is assumed to be known (or inferred from experimental data), while the identity of the intermediate nodes and interactions is unknown. The goal therefore becomes to complete the gap and find a probable subgraph of the PPI network that simultaneously enables signals between all the protein pairs, thereby explaining the overall biological activity. Since the edges in the PPI network can be assigned a probability value (reflecting the credibility of their experimental evidence), by taking the negative log of these values as edge weights, the task becomes minimizing the total edge weight, leading to an instance of the Steiner network problem. We have previously used this approach to study the propagation of a stabilizing signal in pro-inflammatory T cells, leading to the identification of a new molecular pathway (represented by a subgraph of the PPI network) that is critical for mounting an auto-immune response, as validated experimentally by perturbation assays and disease models in  mice \cite{Wu2013tcell}. Variations on this idea have been used successfully by \cite{scott2005identifying}, \cite{huang2009integration}, \cite{yosef2009toward}, \cite{BenSh2012yeast} and others.

While these studies contributed to the understanding of signal transduction pathways in living cells, they ignore a critical aspect of the underlying biological complexity. In reality, proteins (nodes) can become activated or inactivated at different points in time and with different dynamics, thereby giving rise to a PPI network that changes over time \cite{PrzytyckaSS10}. Recent advances in mass-spectrometry based measurements provide a way to estimate these changes at high throughput (\textit{e.g.,} measuring phosphorylation levels or overall protein abundance, proteome-wide) \cite{Kanshin20151202}. The next challenge is therefore to study connectivity problems that take into account not only the endpoints of each demand, but also the time (or condition) in which this demand should be satisfied. This added complication was tackled by Mazza et al. \cite{mazza2014minimum}, who introduced the ``Minimum $k$-Labeling (MKL)'' problem. In this setting, each connectivity demand comes with a label, which represents a certain experimental condition or time point. The task is to label edges in the PPI network so as to satisfy each demand using its respective label, while minimizing the number of edges in the resulting subgraph and the number of labels used to annotate these edges. They give a brief theoretical analysis of MKL, then present an ILP-based algorithm that works well in several experiments. While MKL was an important first step, the challenge remains to satisfy connectivity demands for different conditions, while taking into account changes to the activity of proteins under each condition, thus providing more reliable hypotheses for the mechanism of signal transduction. Simplifying the experimentally-measured activity of proteins to a binary view of presence/absence, the work presented here aims to address this challenge.

Our formulation of the $\TSN$ problem naturally arises from our work on PPI networks: the nodes $V$ represent proteins; the edges $E$ represent protein interactions, weighed by the confidence of the supporting experimental data (for $e \in E$, $w(e)=-\log(p(e))$ where $p(e)$ is the probability associated with interaction $e$); the existence function $\sigma$ can be derived from a proteome-wide assay, e.g., measuring protein abundance or phospnorylation levels over a set of $T$ time points or experimental conditions. Notably, while in this setting we assign states of presence/absence to proteins (nodes), it is mutually polynomial-time reducible with the formulation above (where the existence function is defined over the edges; see \autoref{proposition:variants_mutually_reducible}). The connectivity demands include pairs of proteins $(a,b)$ that represent the known end points of an unknown signal transduction cascade that is active in experimental condition $t \in [T]$. Finally, the desired output is a maximum-probability subgraph of the PPI network $G$ that explains the transduction of signals between all the queried protein pairs in the respective experimental condition, while taking into account the experimentally-derived information of protein presence/absence.
%----------------------------------------------------------------------------------------
%----------------------------------------------------------------------------------------

\end{document}